\documentclass{article}
\usepackage[numbers]{natbib}
\usepackage[english]{babel}
\usepackage{algpseudocode}
\usepackage{makecell}
\usepackage[english]{babel}
\usepackage{array} 
\usepackage{booktabs} 
\usepackage[letterpaper,top=2cm,bottom=2cm,left=3cm,right=3cm,marginparwidth=1.75cm]{geometry}
\usepackage[ruled,linesnumbered,vlined]{algorithm2e}
\usepackage{algpseudocode}
\usepackage{amssymb}
\usepackage{amsmath}
\usepackage{amsthm}
\usepackage{graphicx}
\usepackage[colorlinks=true, allcolors=blue]{hyperref}
\usepackage{natbib}
\usepackage{subcaption}
\usepackage{pgfplots}
\pgfplotsset{compat=1.17}
\usepackage{tikz}
 \usetikzlibrary{patterns}
\usepackage{pgfplotstable}
    \pgfplotsset{
    name nodes near coords/.style={
        every node near coord/.append style={
            name=#1-\coordindex,
            alias=#1-last,
        },
    },
    name nodes near coords/.default=coordnode
    }

\newtheorem{theorem}{Theorem}[section]
\newtheorem{definition}[theorem]{Definition}

\newtheorem{corollary}{Corollary}[theorem]
\newtheorem{lemma}[theorem]{Lemma}
\newtheorem{invariant}[theorem]{Invariant}
\newtheorem{example}[theorem]{Example}

\newcommand{\MPB}{\mathsf{MPB}}
\newcommand{\bX}{\mathbf{X}}

\newcommand{\bc}{\mathbf{c}}
\newcommand{\bp}{\mathbf{p}}

\DeclareMathOperator*{\argmin}{argmin}

\title{Approximately EFX and PO Allocations for Bivalued Chores}
\author{
Zehan Lin\thanks{IOTSC, University of Macau. {\{yc47490,xiaoweiwu,yc17423\}@um.edu.mo}. The authors are ordered alphabetically.}
\and
Xiaowei Wu$^*$
\and 
Shengwei Zhou$^*$
}
\date{}

\begin{document}

\maketitle

\begin{abstract}
    We consider the computation for allocations of indivisible chores that are approximately EFX and Pareto optimal (PO). 
    Recently, Garg et al. (2024) show the existence of $3$-EFX and PO allocations for bi-valued instances, where the cost of an item to an agent is either $1$ or $k$ (where $k > 1$) by rounding the (fractional) earning restricted equilibrium.
    In this work, we improve the approximation ratio to $(2-1/k)$, while preserving the Pareto optimality.
    Instead of rounding fractional equilibrium, our algorithm starts with the integral EF1 equilibrium for bi-valued chores, introduced by Garg et al. (AAAI 2022) and Wu et al. (EC 2023), and reallocates items until approximate EFX is achieved. 
    We further improve our result for the case when $k=2$ and devise an algorithm that computes EFX and PO allocations.
\end{abstract}

\section{Introduction}

Fair allocation has received significant attention over the past few decades in the fields of computer science, economics, and mathematics. 
The problem focuses on allocating a set $M$ of $m$ items among a group $N$ of $n$ agents, where agents have heterogeneous valuation functions for the items. 
When the valuation functions assign positive values, the items are considered as goods, such as resources; when the valuation functions assign negative values, the items are regarded as chores, such as tasks. 
The goal is to compute an allocation of the items that is fair to all agents.
In this work, we focus on the allocation of chores, where each agent has an additive cost function that assigns a non-negative cost to every item, and explore the existence of allocations that are fair and efficient.
 
Envy-freeness (EF) is one of the most well-studied fairness notions. 
An allocation is envy-free if no agent wants to exchange her bundle of items with any other agent. 
Unfortunately, envy-free allocations are not guaranteed to exist for indivisible items. 
Consequently, researchers have focused on the relaxations of envy-freeness. 
Envy-freeness up to one item (EF1) and envy-freeness up to any item (EFX) are two widely studied relaxations. 
Generally speaking, EF1 allocations require that the envy between any two agents can be eliminated by removing some item, while EFX allocations require that the envy can be eliminated by removing any item. 
In addition to fairness, efficiency is another important measure of the quality of allocations. 
The existence of fair and efficient allocations has recently drawn significant attention. 
Unfortunately, efficiency and fairness often compete with each other, and many fair allocations exhibit poor efficiency guarantees. 
Pareto optimality (PO) is one of the most widely used measures for efficiency. 
An allocation is called PO if no other allocation can improve the outcome for one agent without making another agent worse off. 

It has been shown that EF1 allocations always exist for goods~\cite{conf/sigecom/LiptonMMS04}, chores, and the mixture of goods and chores~\cite{journals/aamas/AzizCIW22, conf/approx/BhaskarSV21}. 
There are several simple polynomial-time algorithms for computing EF1 allocations such as envy-cycle elimination and round-robin~\cite{conf/sigecom/LiptonMMS04,conf/approx/BhaskarSV21}.
Compared to the well-established results for EF1 allocations, the existence of EFX allocations remains a major open problem. 
For general cost functions, the divide-and-choose algorithm can be applied to compute EFX allocations when $n = 2$. 
However, it is still unknown whether EFX allocations exist for chore instances with more than two agents. 
So far, EFX allocations for chores are shown to exist for two types of chore~\cite{conf/atal/0001LRS23}, the case of $m \leq 2n$~\cite{garg2024fair, conf/sagt/KobayashiMS23}, and agents with leveled preferences~\cite{journals/corr/abs-2109-08671}.
Aziz et al.~\cite{journals/ai/AzizLMWZ24} show that EFX allocations can be computed for identical ordering (IDO) instances, which is extended by
Kobayashi et al. \cite{conf/sagt/KobayashiMS23} to the case when the cost functions of all but one agent are IDO.
For bi-valued instances (where the cost of any item to any agent can only take one of two fixed values), Zhou and Wu~\cite{journals/ai/ZhouW24} propose a polynomial-time algorithm that computes EFX allocations for three agents.
The result is generalized by Kobayashi et al.~\cite{conf/sagt/KobayashiMS23} to the case of three agents with personalized bi-valued cost functions and improved by Garg et al.~\cite{conf/ijcai/GargMQ23}, who show the existence of EFX and PO allocations for three bi-valued agents.

Given the challenges in computing EFX allocations, many studies focus on the approximations of EFX allocations.
Zhou and Wu~\cite{journals/ai/ZhouW24} demonstrate the existence of a polynomial-time algorithm that computes ($2 + \sqrt{6}$)-approximate EFX (($2 + \sqrt{6}$)-EFX) allocations for three agents. 
Recently, Afshinmehr et al.~\cite{journals/corr/abs-2410-18655} and Christoforidis et al.~\cite{conf/ijcai/ChristoforidisS24} improve this result, showing that $2$-EFX allocations can be computed in polynomial time for three agents. 
For a general number of agents, the state-of-the-art results for approximating EFX allocation are by Garg et al.~\cite{garg2024fair}, who show the existence of $4$-EFX allocations for general additive instances, and $3$-EFX and PO allocations for bi-valued instances.

\subsection{Our Results and Techniques}

In this work, we consider the fair allocation of chores for bi-valued instances, where the cost of any item on any agent is either $1$ or $k$, where $k>1$.
We also call them $\{1,k\}$-instances.

\smallskip

In our first result, we improve the state-of-the-art approximation ratio for EFX and PO allocations for bi-valued chores to $(2 - 1/k)$. 

\medskip
\noindent
{\bf Result 1} (Theorem \ref{theorem:2EFX_PO_bivalued}){\bf .} \label{Result1}
{\em For $\{1, k\}$-instances of indivisible chores, there exists an algorithm that computes ($2 - 1/k$)-EFX and PO allocations in polynomial time.}
	
\medskip

As in recent works~\cite{conf/atal/EbadianP022,conf/aaai/GargMQ22,conf/sigecom/WuZZ23}, our analysis follows the Fisher market analysis framework, a classic and widely used approach for computing allocations that are both fair and PO.
By associating a payment to every chore and computing an equilibrium, we can focus on achieving fairness, where the Pareto optimality will be implied by the equilibrium.
Following this framework, Garg et al.~\cite{garg2024fair} start with an envy-free fractional earning restricted equilibrium (where the earning of an agent from any item is bounded) and devise a rounding scheme that computes a $3$-EFX and PO allocation\footnote{As an intermediate step, they first round the equilibrium to a $2$-EF2 and PO allocation.}.
Instead, in this work we start with the pEF1 integral equilibrium for bi-valued chores~\cite{conf/aaai/GargMQ22,conf/sigecom/WuZZ23}, and design an algorithm that executes a sequence of item reallocations and returns a $(2-1/k)$-EFX and PO allocation in polynomial time.
A crucial difference between our analysis and that of Garg et al.~\cite{garg2024fair} is that since we start with a pEF1 equilibrium, to compute a $(2-1/k)$-EFX and PO allocation, it suffices to consider the case when every agent possesses at most one high payment item.
Therefore, by partitioning the agents into two groups depending on whether they have high payment items, we show that it suffices to reallocate items across the two groups.
Since our algorithm and analysis do not require rounding fractional items, it is much simpler than that of Garg et al.~\cite{garg2024fair}.

\smallskip

Moreover, we show that better results can be obtained for the case when $k=2$.

\medskip
\noindent
{\bf Result 2} (Theorem \ref{theorem:EFX_k=2}){\bf .}
{\em For $\{1,2\}$-instances of indivisible chores, there exists an algorithm that computes EFX and PO allocations in polynomial time.}
\medskip

We note that very few results are known for EFX and PO allocations for chores.
In fact, even the existence of EF1 and PO allocations for chores remains a major open problem, which is in sharp contrast to the allocation of goods~\cite{journals/teco/CaragiannisKMPS19}.
The existence of EFX and PO allocations for chores has only been established for some special cases, e.g., 
for three bi-valued agents~\cite{conf/ijcai/GargMQ23},
for bi-valued instances with $m\leq 2n$~\cite{garg2024fair} and for binary instances ($\{0,1\}$-instances)~\cite{journals/corr/abs-2308-12177}.
For some slightly more general cases, it has been shown that EFX is incompatible with Pareto optimality.
For example, Garg et al.~\cite{conf/ijcai/GargMQ23} show that EFX and PO allocations do not always exist for the case with three agents, and the case with two types of chores.
Tao et al.~\cite{journals/corr/abs-2308-12177} demonstrate that EFX and PO allocations do not always exist for ternary instances, where the cost of any item on any agent can only take values in $\{0, 1, 2\}$.

Our algorithm and analysis follow a similar framework as our first result, e.g., we start with the pEF1 equilibrium and reallocate items to strengthen the fairness guarantee.
However, to improve the approximation ratio further, we need a finer classification of the agents into groups. 
Fortunately, for the pEF1 equilibrium for $\{1,2\}$-instances, it can be shown that the earning of agents can only take values in $\{z,z+1,z+2\}$, for some integer $z$.
Therefore, it suffices to reallocate items between agents with earning $z$ and those with earning $z+2$, until one of the two groups becomes empty, or the allocation becomes EFX.
By proving that every agent with earning $z$ will participate in item reallocation at most once, we show that our algorithm returns an EFX and PO allocation in polynomial time.

\subsection{Other Related Work}

Due to the vast literature on the fair allocation problem, in the following we only review the results regarding the approximation and computation of EFX allocations.
For a comprehensive overview of other related works, please refer to the surveys by Aziz et al.~\cite{journals/sigecom/AzizLMW22}, Amanatidis et al.~\cite{journals/ai/AmanatidisABFLMVW23}.

\paragraph{EFX Allocations for Goods.}
Compared to the allocation of chores, the case of goods admits more fruitful results.
The EFX allocations are shown to exist for two agents with arbitrary valuations or any number of agents with identical valuations by Plaut et al.~\cite{journals/siamdm/PlautR20}; for three agents by Chaudhury et al.~\cite{journals/jacm/ChaudhuryGM24} and Akrami et al.~\cite{conf/sigecom/AkramiACGMM23}; and for two types of agents by Mahara~\cite{journals/dam/Mahara23}.
For bi-valued instances, Amanatidis et al.~\cite{journals/tcs/AmanatidisBFHV21} and Garg et al.~\cite{journals/tcs/GargM23} establish the existence and computation of EFX and PO allocations, respectively.
Regarding the approximation of EFX allocations, Chan et al.~\cite{conf/ijcai/Chan00W19} demonstrate the existence of $0.5$-EFX allocations, while Amanatidis et al.~\cite{journals/tcs/AmanatidisMN20} improve the approximation ratio to $0.618$. 
Recently, Amanatidis et al.~\cite{conf/sigecom/AmanatidisFS24} improve this ratio further to $2/3$ for some special cases, e.g., seven agents or tri-valued goods.

\paragraph{EFX Allocations with Unallocated Items.}
For the allocation of goods, Chaudhury et al.~\cite{journals/siamcomp/ChaudhuryKMS21} show that EFX allocations exist
if we are allowed to leave at most $(n - 1)$ items unallocated. 
The result is improved to $(n - 2)$ items by Berger et al.~\cite{conf/aaai/BergerCFF22}, who also show that EFX allocations with at most one unallocated item can be computed for four agents. 
For the allocation of chores, EFX allocations with at most $(n - 1)$ unallocated items have been shown to exist for bi-valued instances by Zhou and Wu~\cite{journals/ai/ZhouW24}.

\section{Preliminary}

We consider how to fairly allocate a set of $m$ indivisible items (chores) $M$ to a group of $n$ agents $N$. 
We call a subset of items, e.g. $X \subseteq M$, a bundle. 
Each agent $i \in N$ has an additive cost function $c_i:2^M \to \mathbb{R}^+ \cup \{0\}$ that assigns a cost to every bundle of items. 
For convenience, we use $c_i(e)$ to denote $c_i(\{e\})$, the cost of agent $i \in N$ on item $e \in M$, thus $c_i(X) = \sum_{e \in X} c_i(e)$ for all $X \subseteq M$.
We use $\bc = (c_1, \dots, c_n)$ to denote the cost functions of agents. 
For any subset of items $X\subseteq M$ and item $e\in M$, we use $X+e$ and $X-e$ to denote $X\cup \{e\}$ and $X\setminus \{e\}$, respectively.
An allocation $\bX = (X_1, \dots, X_n)$ is an $n$-partition of the items $M$ such that $X_i \cap X_j = \emptyset$ for all $i \neq j$ and $\cup_{i \in N}X_i = M$, where agent $i$ receives bundle $X_i$. 
Given an instance $\mathcal{I} = (N, M, \bc)$, our goal is to find an allocation $\bX$ that is fair to all agents and also efficient.

\begin{definition}[EF]
An allocation $\bX$ is envy-free (EF) if for any $i, j \in N$ we have $c_i(X_i) \leq c_i(X_j)$.
\end{definition}

\begin{definition}[EF1]
An allocation \textbf{X} is called envy-free up to one item (EF1) if for any agents $i, j \in N$, either $X_i = \emptyset$, or there exists an item $e \in X_i$ such that $c_i(X_i - e) \leq c_i(X_j)$.
\end{definition}

\begin{definition}[$\beta$-EFX]
For any $\beta \geq 1$, an allocation $\bX$ is $\beta$-approximate envy-free up to any item ($\beta$-EFX) if for any agents $i, j \in N$, either $X_i = \emptyset$, or for any item $e \in X_i$, we have
\begin{equation*}
    c_i(X_i - e) \leq \beta \cdot c_i(X_j).
\end{equation*}
When $\beta = 1$, the allocation is EFX.   
\end{definition}

Notice that all the fairness notions mentioned above are scale-free, i.e., if an allocation satisfies one of these notions, then it remains to satisfy the same notion if we rescale any cost function.

\begin{definition}[PO]
An allocation $\bX'$ \emph{Pareto dominates} another allocation $\bX$ if $c_i(X_i') \leq c_i(X_i)$ for all $i \in N$ and the inequality is strict for at least one agent. 
An allocation X is \emph{Pareto optimal} (PO) if $\bX$ is not dominated by any other allocation.
\end{definition}

\begin{definition}[Bi-valued Instances] 
An instance is called bi-valued if there exist constants $a, b \geq 0$ such that for any agent $i \in N$ and item $e \in M$, we have $c_i(e) \in \{a, b\}$.
\end{definition}

We remark that when $a = 0$ or $b = 0$, such instances reduce to binary instances, for which EFX and PO allocations exist and can be computed in polynomial time~\cite{journals/corr/abs-2308-12177}.
Hence, in the following, we assume w.l.o.g. that $a \neq 0$ and $b\neq 0$.
Additionally, we can rescale all the cost functions such that $c_i(e) \in \{1, k\}$, where $k > 1$.
Note that for all $i\in N$, there exists at least one item $e \in M$ with cost $c_i(e) = 1$ as otherwise we can rescale the cost functions so that $c_i(e) = 1$ for all $e\in M$.
For convenience, we refer to bi-valued instances mentioned above as $\{1, k\}$-instances. 

\paragraph{Fisher Market.}
In the Fisher market, there is a payment vector $\bp$ that assigns each chore $e \in M$ a payment $p(e) > 0$.
For any subset $X \subseteq M$, let $p(X) = \sum_{e \in X} p(e)$.
Given the payment vector $\bp$, we define the pain-per-buck ratio $\alpha_{i, e}$ of agent $i$ for chore $e$ as $\alpha_{i, e} = \frac{c_i(e)}{p(e)}$, and the minimum pain-per-buck (MPB) ratio $\alpha_{i}$ of agent $i$ to be $\alpha_{i} = \min_{e \in M} \{\alpha_{i, e}\}$. 
Intuitively, when agent $i$ receives chore $e$, she needs to perform $c_i(e)$ units of work and receives $p(e)$ units of payment in return.
Therefore every agent would like to receive items with low cost and high payment.
Note that while the costs are subjective, the payments are objective.
For each agent $i$, we define $\MPB_i = \{e \in M : \alpha_{i, e} = \alpha_i\}$ as the set of items with the minimum pain-per-buck ratios and we call each item $e \in \MPB_i$ an MPB item of agent $i$. 
An allocation $\bX$ with payment $\bp$ forms a (Fisher market) equilibrium $(\bX, \bp)$ if each agent only receives her MPB chores, i.e. $\forall i \in N, X_i \subseteq \MPB_i$.

\medskip

Following the First Welfare Theorem~\cite{mas1995microeconomic}, for any market equilibrium $(\bX,\bp)$, the allocation $\bX$ is PO.
For completeness, we provide a short proof here.

\begin{lemma}
    For any equilibrium $(\bX, \bp)$, the allocation $\bX$ is PO.
\end{lemma}
\begin{proof}
    If the allocation $\bX$ with payment $\bp$ is an equilibrium, then for any agent $i\in N$, any item $e\in X_i$ and any agent $j\neq i$ we have 
    \begin{equation*}
        \frac{c_i(e)}{\alpha_i} = \frac{c_i(e)}{\alpha_{i,e}} = p(e) = \frac{c_j(e)}{\alpha_{j,e}} \le \frac{c_j(e)}{\alpha_{j}}.
    \end{equation*}
    
    Thus the allocation minimizes the objective $\sum_{i\in N} \frac{c_i(X_i)}{\alpha_i}$. 
    Any Pareto improvement would strictly decrease this objective, which leads to a contradiction.
    So the allocation $\bX$ is PO.
\end{proof}
  
\begin{definition}[pEF1]
An equilibrium $(\bX, \bp)$ is called payment envy-free up to one item (pEF1) if for any $i, j \in N$, either $X_i = \emptyset$, or there exists an item $e \in X_i$, such that $p(X_i - e) \leq p(X_j)$.
\end{definition}

\begin{definition}[$\beta$-pEFX]
For any $\beta \geq 1$, an equilibrium $(\bX, \bp)$ is called $\beta$-approximate envy-free up to any item ($\beta$-EFX), if for any two agents $i, j \in N$ either $X_i = \emptyset$ or for any item $e \in X_i$, we have
\begin{equation*}
    p(X_i - e) \leq \beta \cdot p(X_j).
\end{equation*}
when $\beta$ = 1, the equilibrium is pEFX.
\end{definition}

Next, we establish a few lemmas showing that the approximate payment envy-freeness implies approximate envy-freeness.
Note that vice-versa is not true, e.g., it is possible that an allocation is EFX but not pEFX.
In other words, the payment envy-freeness is strictly stronger than envy-freeness.

\begin{lemma} \label{lemma:pEF1_implies_EF1}
If an equilibrium $(\bX, \bp)$ is pEF1, then the allocation $\bX$ is EF1.
\end{lemma}
\begin{proof}
Since $(\bX, \bp)$ is pEF1, for any $i, j \in N$, there exists an item $e \in X_i$, such that $p(X_i - e) \leq p(X_j)$. Therefore, according to the definition of MPB allocation, we have:
\begin{equation*}
    c_i(X_i - e) = \alpha_i \cdot p(X_i - e) \leq \alpha_i \cdot p(X_j) \leq c_i(X_j),
\end{equation*}
where the last inequality holds since the pain-per-buck ratio of agent $i$ on any item is at least $\alpha_i$.
\end{proof}

\begin{lemma} \label{lemma:bpEFX_implies_bEFX}
    Given any equilibrium $(\bX, \bp)$, if agent $i\in N$ is $\beta$-pEFX towards agent $j\in N$, then $i$ is $\beta$-EFX towards $j$. 
\end{lemma}
\begin{proof}
Since $i$ is $\beta$-pEFX towards $j$, for any $e \in X_i$, we have $p(X_i - e) \leq \beta \cdot p(X_j)$.
Given that $(\bX, \bp)$ is an equilibrium, we have:
\begin{equation*}
    c_i(X_i - e) = \alpha_i \cdot p(X_i - e) \leq \beta \cdot \alpha_i \cdot p(X_j) \leq \beta \cdot c_i(X_j). \qedhere
\end{equation*}
\end{proof}

\begin{corollary}\label{corollary:notbpEFX_implies_notbpEFX}
    Given any equilibrium $(\bX, \bp)$, if agent $i$ is not $\beta$-EFX towards $j$, then $i$ is not $\beta$-pEFX towards $j$.
\end{corollary}

\section{\texorpdfstring{$(2 - 1/k)$}{}-EFX and PO Allocations for \texorpdfstring{$\{1,k\}$}{}-Instances} \label{Section3}

In this section we present a polynomial-time algorithm for computing $(2 - 1/k)$-EFX and PO allocations for $\{1,k\}$-valued instances.
For bi-valued instances, several works showed that EF1 and PO allocations exist and can be computed in polynomial time~\cite{conf/atal/EbadianP022,conf/aaai/GargMQ22,conf/sigecom/WuZZ23} based on the Fisher market framework.
For any given bi-valued instance, they compute a pEF1 equilibrium, which implies an EF1 and PO allocation.
Specifically, the equilibrium has the following properties.

\begin{lemma}[\cite{conf/aaai/GargMQ22,conf/sigecom/WuZZ23}]
\label{lemma:1_k_payment_equilibrium}
    There exists an algorithm that given any $\{1,k\}$-instance $(N,M,\bc)$ computes an equilibrium $(\bX, \bp)$ that is pEF1, and satisfies $p(e) \in \{1,k\}$ for all $e\in M$ in $O(nm^2)$ time.
\end{lemma}

For an equilibrium satisfying $p(e) \in \{1,k\}$ for all $e\in M$, we call it a $\{1,k\}$-payment equilibrium.

\subsection{Properties of \texorpdfstring{$\{1,k\}$}{}-Payment Equilibrium}

Depending on the payment vector $\bp$, we divide the items into two groups $L$ and $H$, where $L$ contains all items with low payment and $H$ contains all items with high payment.
We further categorize agents according to whether they receive high payment items.

\begin{definition}
    Chores are categorized as:
    \begin{itemize}
        \item $L = \{e \in M: p(e) = 1\}$, 
        \item $H = \{e \in M: p(e) = k\}$.
    \end{itemize}
    
    Agents are categorized as:
    \begin{itemize}
        \item $N_L = \{i \in N: X_i \subseteq L\}$,
        \item $N_H = \{i \in N: |X_i \cap H| \geq 1\}$.
    \end{itemize}
\end{definition}

We can assume that both $L$ and $H$ are non-empty, as otherwise, all items have the same payment and the equilibrium is already pEFX.
Therefore $N_H$ is non-empty.
We can further assume that $N_L$ is non-empty because otherwise the equilibrium is $(2-1/k)$-pEFX: for all agents $i,j\in N = N_H$ and any $e\in X_i$, we have
\begin{equation*}
    p(X_i - e) \leq p(X_j) + k - 1 = (1 + \frac{k-1}{p(X_j)})\cdot p(X_j) \leq (2-\frac{1}{k})\cdot p(X_j),
\end{equation*}
where the first inequality holds since $(\bX,p)$ is pEF1 and the second inequality holds since $X_j\cap H \neq \emptyset$.

Next we establish some useful properties for $\{1,k\}$-payment equilibrium $(\bX, \bp)$.
Note that these properties hold for any $\{1,k\}$-payment equilibrium, regardless of whether it is pEF1 or not.

\begin{lemma}\label{lemma:1k_equilibrium_properties}
For any $\{1,k\}$-payment equilibrium $(\bX, \bp)$, we have the following properties:
\begin{enumerate}
    \item For any agent $i \in N$, we have $\alpha_i \in \{1, \frac{1}{k}\}$;
    \item For any agent $i\in N_L$, we have $\alpha_i = 1$;
    \item For any agent $i \in N_L$ and item $e \in H$, we have $c_i(e) = k$ and $e \in \MPB_i$.
\end{enumerate}
\end{lemma}
\begin{proof}
We prove these properties one by one.
\begin{itemize}
    \item 
    Fix any agent $i\in N$. Since $c_i(e)\in\{1,k\}$ and $p(e)\in\{1,k\}$ for all $e\in M$, we have $\alpha_{i,e}\in \{ k,1,\frac{1}{k} \}$.
    As there exists at least one item $e$ with $c_i(e) = 1$, we have $\alpha_i \leq \alpha_{i,e} = \frac{1}{p(e)} \leq 1$, which implies the first property.
    
    \item 
    Fix any $i\in N_L$.
    Since $p(e) = 1$ for all $e\in X_i$, we have $\alpha_i = \frac{c_i(e)}{p(e)} = c_i(e) \geq 1$.
    Combining this with the first property, we have $\alpha_i = 1$.
        
    \item 
    Fix any agent $i\in N_L$ and item $e\in H$, we have $c_i(e) = \alpha_{i,e}\cdot p(e) \geq \alpha_i\cdot k = k$.
    Therefore we have $c_i(e) = k$ and $\alpha_{i,e} = 1 = \alpha_i$, which implies $e\in \MPB_i$.
    \qedhere
\end{itemize}
\end{proof}

Note that while it is possible that $\alpha_i = \frac{1}{k}$ for some $i\in N_H$, it happens only if $c_i(e) = 1$ and $p(e) = k$ for all $e\in X_i$.
Moreover, all low payment items are not in $\MPB_i$.

\smallskip

Next, we establish more properties for pEF1 $\{1,k\}$-payment equilibrium.

\begin{lemma}\label{lemma:size-of-lowpaymentitems}
For any $\{1,k\}$-payment equilibrium $(\bX, \bp)$, if $(\bX, \bp)$ is pEF1 for agent $i\in N_H$ (i.e., $i$ is pEF1 towards all other agents), then we have $|X_i \cap L| \leq \min_{j\in N_L}\{|X_j|\}$.
\end{lemma}
\begin{proof}
    Assume otherwise that $|X_i \cap L| > |X_j|$ for some $j\in N_L$.
    Since $|X_i \cap H| > 0$, we have 
    \begin{equation*}
        \min_{e \in X_i}\{p(X_i - e)\} = p(X_i) - k \geq |X_i \cap L| > |X_j| = p(X_j),
    \end{equation*}
    which implies that $i$ is not pEF1 towards $j$ and is a contradiction.
\end{proof}

\begin{lemma}\label{lemma:payment<k}
    For any $\{1, k\}$-payment equilibrium $(\bX, \bp)$, if there exists an agent $i$ that is pEF1 but not $(2 - 1/k)$-pEFX towards another agent $j$, then we have $p(X_j) < k$.
\end{lemma}
\begin{proof}
    Since $i$ is pEF1 towards $j$, we have $p(X_i) \leq p(X_j) + k$.
    Suppose that $p(X_j) \geq k$.
    We have
    \begin{equation*}
        p(X_i) - 1 \leq p(X_j) + k - 1 \leq p(X_j) + (1 - \frac{1}{k}) \cdot p(X_j) = (2 - \frac{1}{k}) \cdot p(X_j).
    \end{equation*}
    In other words, $i$ is $(2 - 1/k)$-pEFX towards agent $j$, which leads to a contradiction.
\end{proof}

\subsection{The Reallocation Algorithm}

Now we are ready to introduce our algorithm.
We begin with the $\{1,k\}$-payment pEF1 equilibrium $(\bX^0, \bp)$ as described in Lemma~\ref{lemma:1_k_payment_equilibrium}.
As long as the allocation is not $(2 - 1/k)$-EFX, we find an agent $i$ that is not $(2 - 1/k)$-EFX towards an agent $j$, and reallocate some items between $i$ and $j$ (see Algorithm~\ref{alg:reallocation_bivalued}).
During these reallocations, we do not change the payment of any item.
Throughout the execution of the algorithm, we maintain the property that $(\bX, \bp)$ is an equilibrium.
Specifically, we ensure that when we reallocate an item $e$ to agent $i$, we have $e\in \MPB_i$.

Note that if $\bX^0$ is not $(2 - 1/k)$-EFX, then $(\bX^0,\bp)$ is not $(2 - 1/k)$-pEFX.
We first show that if $(\bX^0,\bp)$ is not $(2 - 1/k)$-pEFX, we have some useful properties regarding the earnings of agents.

\begin{lemma}\label{lemma:b<k}
    If the equilibrium $(\bX^0, \bp)$ is not $(2 - 1/k)$-pEFX, then we have $\min_{i\in N_L}\{ p(X^0_i) \} < k$.
\end{lemma}

\begin{proof}
    Since $(\bX^0,\bp)$ is pEF1 and not $(2 - 1/k)$-pEFX, there exist $i, j\in N$ such that agent $i$ is not $(2 - 1/k)$-pEFX towards $j$.
    Following Lemma~\ref{lemma:payment<k}, we have $p(X^0_j) < k$.
    Recall that every agent in $N_H$ has at least one high payment item.
    Hence we have $j\in N_L$ and $\min_{i\in N_L}\{ p(X^0_i) \} \leq p(X^0_j) < k$.
\end{proof}

Given the above lemma and that $(\bX^0,\bp)$ is pEF1, we have the following property for $(\bX^0,\bp)$.

\begin{corollary} \label{corollary:non2pEFX_property}
    If $(\bX^0, \bp)$ is not $(2 - 1/k)$-pEFX, then there exists an integer $z < k$ such that
    \begin{itemize}
        \item for all agent $j\in N_L$, $p(X_j) \in \{z,z+1\}$;
        \item for all agent $i\in N_H$, $p(X_i) \in [k, k+z]$.
    \end{itemize}
\end{corollary}

Corollary~\ref{corollary:non2pEFX_property} highlights that our starting point, the pEF1 equilibrium, possesses several useful properties, e.g., agents in $N_L$ are pEFX towards each other and agents in $N_H$ are $(2 - 1/k)$-pEFX towards each other.
Therefore, to compute an allocation that is $(2 - 1/k)$-EFX and PO, it suffices to eliminate the strong envy from agents in $N_H$ to agents in $N_L$.

\paragraph{High Level Idea.}
The most natural idea is to reallocate items from agents in $N_H$ to those in $N_L$, until the allocation is $(2-1/k)$-pEFX (which implies $(2-{1}/{k})$-EFX and PO, by Lemma~\ref{lemma:bpEFX_implies_bEFX}).
However, such reallocation must be ``MPB-feasible'', i.e., whenever we allocate an item $e$ to agent $i$, we need to ensure that $e\in \MPB_i$, as otherwise the equilibrium property ceases to hold.
As we have shown in Lemma~\ref{lemma:1k_equilibrium_properties}, all high payment items are MPB to agents in $N_L$.
However, reallocating these items does not help achieving $(2-1/k)$-pEFX.
On the other hand, there might not be sufficient low payment items from $N_H$ that can be reallocated to $N_L$, and it is difficult to maintain $(2-1/k)$-pEFX between agents in $N_L$.
To get around these difficulties, we instead focus on ensuring $(2 - 1/k)$-EFX, while guaranteeing that all reallocations are MPB-feasible (which maintains the equilibrium and the PO property).
We establish a useful property in Lemma~\ref{lemma:property_of_not2EFX} that if agent $i$ is not $(2-1/k)$-EFX towards agent $j$, then we have $X_j \subseteq \MPB_i$.
That means, while we cannot guarantee the MPB-feasibility for allocating the low payment items from $X_i$ to $X_j$, we can instead \emph{exchange} all low payment items in $X_j$ with the high payment item in $X_i$, which preserves the equilibrium property and balances the earning between agents $i$ and $j$.
Note that throughout the whole execution of the algorithm, we do not change the payment of any item.

\medskip

In the following, we say that agent $i$ \emph{strongly envies} agent $j$ if $i$ is not $(2 - 1/k)$-EFX towards $j$.

\begin{algorithm}[htbp]
\caption{$(2 - 1/k)$-EFX and PO allocations for $\{1,k\}$-instances}
\label{alg:reallocation_bivalued}
\KwIn{A pEF1 $\{1,k\}$-payment equilibrium $(\bX, \bp)$ for $\{1,k\}$-instances}
\While{exists agent $i$ that is not $(2 - 1/k)$-EFX towards agent $j$}{
    $j \leftarrow \argmin\{ p(X_{j'}): i \text{ strongly envies } j' \}$; \qquad\ \ \tcp{$i\in N_H$ and $j\in N_L$, by Lemma~\ref{lemma:property_of_not2EFX}}
    let $e \in X_i \cap H$ be the high payment item in $X_i$\;
    update $X_i \leftarrow X_i \cup X_j - e$; \qquad\quad\qquad\qquad\qquad\ \tcp{$X_j\in \MPB_i$, by Lemma~\ref{lemma:property_of_not2EFX}}
    update $X_j \leftarrow \{e\}$; \quad\qquad\qquad\qquad\qquad\qquad\qquad \tcp{$e\in \MPB_j$, by Lemma~\ref{lemma:1k_equilibrium_properties}}
    $N_H \leftarrow N_H \setminus \{i\} \cup \{j\}$\;
    $N_L \leftarrow N_L \setminus \{j\} \cup \{i\}$\;
    }
\KwOut{Allocation $\bX$}
\end{algorithm}

We use $\bX^t$ to denote the allocation at the beginning of round $t$. During the item reallocations, we may reallocate high-payment items, and thus agents may shift between $N_L$ and $N_H$.
We use $N^t_H$ and $N^t_L$ to denote the set of agents with and without high payment items at the beginning of round $t$, respectively.
Specifically, $N^0_L, N^0_H$ denote the responding agent sets in the initial equilibrium $(\bX^0, \bp)$.
As an illustration, we give an example of the execution of Algorithm~\ref{alg:reallocation_bivalued} in Example~\ref{example:execution_of_reallocation}.

\begin{example}\label{example:execution_of_reallocation}
    Consider the following $\{1, k\}$-instance with seven agents and $k = 6$.
    At the beginning of round $t$, $p(X_1^{t}) = 2$, $p(X_2^{t}) = 2$, $p(X_3^{t}) = 3$, $p(X_4^{t}) = 6$, $p(X_5^{t}) = 7$, $p(X_6^{t}) = 8$, and $p(X_7^{t}) = 8$. In this instance, agents $1$, $2$, and $3$ are in group $N_L^{t}$ while agents $4$, $5$, $6$, and $7$ are in group $N_H^t$. Agent $7$ is not $(2 - 1/k)$-EFX towards some agents, with agent $1$ having the minimum earning among them. Then, following Algorithm \ref{alg:reallocation_bivalued}, agent $1$ receives $X_7^{t} \cap H$ and agent $7$ receives all items in $X_1^{t}$. See the illustration of the example in Figure~\ref{figure:exchange-H-L}.
\end{example}

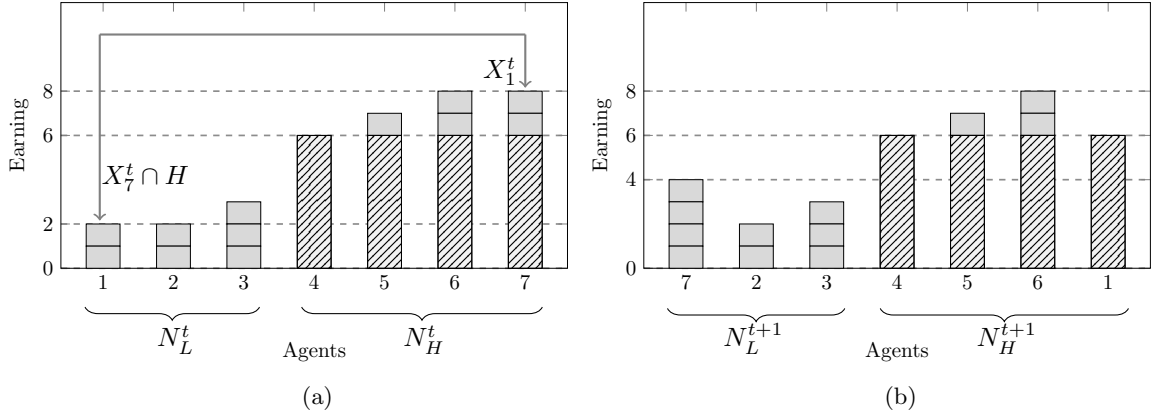
\begin{figure}[htb]
    \centering
    \begin{subfigure}[c]{0.45\textwidth} 
        \centering
        \begin{tikzpicture}[scale=0.8]
        \begin{axis}[
            ybar stacked,
            ymax=12,
            ymin=0,
            enlarge x limits=0.1,
            enlarge y limits=false,
            bar width=16pt,
            height=6cm, width=10cm,
            ylabel={Earning},
            ytick={0,2,6,8},
            xlabel={Agents},
            xlabel style={at={(0.5, -0.25)}, anchor=north},
            symbolic x coords={1, 2, 3, 4, 5, 6, 7}, 
            xtick=data,
            xticklabels={1, 2, 3, 4, 5, 6, 7},
            ymajorgrids=true,
            grid style={dashed, line width=0.8pt, gray!90},
        ]
        \addplot+[draw=black,postaction={pattern=north east lines}, fill=gray!10] coordinates {
            (1,0) (2,0) (3,0) (4,6) (5,6) (6,6) (7,6)
        };
        \addplot+[draw=black, fill=gray!30] coordinates {
            (1,1) (2,1) (3,1) (4,0) (5,0) (6,0) (7,0)
        }; 
        \addplot+[draw=black, fill=gray!30] coordinates {
            (1,1) (2,1) (3,1) (4,0) (5,0) (6,0) (7,0)
        }; 
        \addplot+[draw=black, fill=gray!30] coordinates {
            (1,0) (2,0) (3,1) (4,0) (5,0) (6,0) (7,0)
        }; 
        \addplot+[draw=black, fill=gray!30] coordinates {
            (1,0) (2,0) (3,0) (4,0) (5,1) (6,1) (7,1)
        };
        \addplot+[draw=black, fill=gray!30] coordinates {
            (1,0) (2,0) (3,0) (4,0) (5,0) (6,1) (7,1)
        };
        \end{axis}
        
        \node[anchor=north] at (rel axis cs:0.31, -0.18) {\(N_L^{t}\)};
        \node[anchor=north] at (rel axis cs:0.80, -0.18) {\(N_H^{t}\)};
        
        \draw[<-, thick, gray] (rel axis cs:0.16,0.18) -- (rel axis cs:0.16,0.88);
        \draw[-, thick, gray] (rel axis cs:0.16,0.88) -- (rel axis cs:1.0,0.88);
        \draw[->, thick, gray] (rel axis cs:1.0,0.88) -- (rel axis cs:1.0,0.68);
        \node[anchor=north] at (rel axis cs:0.95,0.83) {\(X_1^{t}\)};
        \node[anchor=north] at (rel axis cs:0.25, 0.43) {\(X_7^{t} \cap H\)};
        
        \draw [decorate,decoration={brace,amplitude=5pt,mirror,raise=1pt},yshift=0pt]
        (0.4,-0.6) -- (3.3,-0.6);
        \draw [decorate,decoration={brace,amplitude=5pt,mirror,raise=1pt},yshift=0pt]
        (4.0,-0.6) -- (8.0,-0.6);
        \end{tikzpicture}
        \caption{\hspace{-1.47cm}     }
        \end{subfigure}
        \hspace{0.5cm}
    \begin{subfigure}[c]{0.45\textwidth} 
        \centering
        \begin{tikzpicture}[scale=0.8]
        \begin{axis}[
            anchor=north,
            yshift=-3cm,
            ybar stacked,
            ymax=12,
            ymin=0,
            enlarge x limits=0.1,
            enlarge y limits=false,
            bar width=16pt,
            height=6cm,
            width=10cm,
            ylabel={Earning},
            ytick={0, 4, 6, 8},
            xlabel={Agents},
            xlabel style={at={(0.5, -0.25)}, anchor=north},
            symbolic x coords={1,2,3,4,5,6,7},
            xtick=data,
            xticklabels={7, 2, 3, 4, 5, 6, 1},
            ymajorgrids=true,
            grid style={dashed, line width=0.8pt, gray!90}
        ]
        \addplot+[draw=black, fill=gray!30]coordinates {
            (1,1) (2,1) (3,1) (4,0) (5,0) (6,0) (7,0)
        };
        \addplot+[draw=black, fill=gray!30]coordinates {
            (1,1) (2,1) (3,1) (4,0) (5,0) (6,0) (7,0)
        };
        \addplot+[draw=black, fill=gray!30]coordinates {
            (1,1) (2,0) (3,1) (4,0) (5,0) (6,0) (7,0)
        };
        \addplot+[draw=black, fill=gray!30]coordinates {
            (1,1) (2,0) (3,0) (4,0) (5,0) (6,0) (7,0)
        };
        \addplot+[draw=black,postaction={pattern=north east lines}, fill=gray!10]coordinates {
            (1,0) (2,0) (3,0) (4,6) (5,6) (6,6) (7,6)
        };
        \addplot+[draw=black, fill=gray!30]coordinates {
            (1,0) (2,0) (3,0) (4,0) (5,1) (6,1) (7,0)
        };
        \addplot+[draw=black, fill=gray!30]coordinates {
            (1,0) (2,0) (3,0) (4,0) (5,0) (6,1) (7,0)
        };
        \end{axis}
        
        \node[anchor=north] at (rel axis cs:-0.20,-1.85) {\(N_L^{t + 1}\)};
        \node[anchor=north] at (rel axis cs:0.30,-1.85) {\(N_H^{t + 1}\)};
        
        \draw [decorate,decoration={brace,amplitude=5pt,mirror,raise=1pt},yshift=0pt]
        (rel axis cs: -0.37,-1.82) -- (rel axis cs: -0.02,-1.82);
        \draw [decorate,decoration={brace,amplitude=5pt,mirror,raise=1pt},yshift=0pt]
        (rel axis cs: 0.05,-1.82) -- (rel axis cs: 0.54,-1.82);
        \end{tikzpicture}
        \caption{\hspace{-1.47cm}  }
    \end{subfigure}
    \caption[c]{The illustration of Example~\ref{example:execution_of_reallocation}. ($a$) The earning status of agents at the beginning of round $t$. ($b$) The earning status of agents after the swap operation.}
    \label{figure:exchange-H-L}
\end{figure}

\subsection{Invariants and Analysis of Algorithm}

Note that we can use $\bX^{t + 1}$, $N_H^{t + 1}$, $N_L^{t + 1}$ to denote the allocation, the set of agents with and without high payment items at the end of round $t$, respectively.
Before analyzing the algorithm, we first introduce several invariants that will be useful for our subsequent analysis. In our algorithm, we ensure that the equilibrium property is never violated, i.e., $(\bX^t, \bp)$ is an equilibrium for all $t$, which guarantees that the final allocation is PO. 

\begin{invariant}[Equilibrium Invariant] \label{invariant:1kequilibrium}
    For all $t\geq 0$, $(\bX^t, \bp)$ is a $\{1,k\}$-payment equilibrium.
\end{invariant}

In addition, our algorithm maintains the invariant that agents within the same group are always $(2 - 1/k)$-EFX towards each other. This means that we can focus solely on eliminating envy between different groups. Furthermore, their earnings are always bounded within the specified ranges.

\begin{invariant}[Low Group Invariant] \label{invariant:2EFX_for_N_L}
    For all $t\geq 0$, agents in $N^t_L$ are $(2 - 1/k)$-EFX towards each other.
    Moreover, for all $j\in N_L$, we have $p(X^t_j) \in [z, k+z]$.
\end{invariant}

\begin{invariant}[High Group Invariant] \label{invariant:2pEFX_for_N_H}
    For all $t\geq 0$, $i \in N_H^t$, we have $p(X^t_i) \in [k, k+z]$.
\end{invariant}

Note that Invariant~\ref{invariant:2pEFX_for_N_H} implies that the allocation is pEF1 for agents in $N_H^t$.
Moreover, agents in $N_H^t$ are $(2 - 1/k)$-pEFX towards each other (which is stronger than being $(2 - 1/k)$-EFX), because for any $i,j\in N_H^t$ we have (recall that $z<k$)
\begin{equation*}
    p(X^{t}_i) - 1 \leq k + z - 1 < 2k - 1 \leq (2 - \frac{1}{k})\cdot p(X^{t}_j).
\end{equation*}


In the following, we show that if the invariants are maintained (for $(\bX^{t},\bp)$, $N^{t}_L$ and $N^{t}_H$) when round $t$ begins, then they remain to hold when round $t$ ends.
Particularly, all invariants hold when round $1$ begins, for $(\bX^0,\bp)$, $N^0_L$ and $N^0_H$.

\begin{lemma} \label{lemma:property_of_not2EFX}
    If there exists an agent $i$ that is not $(2 - 1/k)$-EFX towards another agent $j$ in round $t$, then we have $i\in N^{t}_H$, $j\in N^{t}_L$, $\alpha_i = 1$, and $X^{t}_j \subseteq \MPB_i$.
\end{lemma}
\begin{proof}
    By Invariant~\ref{invariant:2EFX_for_N_L} and \ref{invariant:2pEFX_for_N_H}, if $i$ is not $(2 - 1/k)$-EFX towards $j$, then they must be from different groups.
    Moreover, if $i \in N^{t}_L$ and $j \in N^{t}_H$, then we have 
    $$p(X_i^{t}) - 1 \leq k + z - 1 < 2k - 1 \leq (2 - \frac{1}{k}) \cdot p(X_j^{t}), 
    $$
    where the first inequality holds since $p(X_i^{t}) \in [z, k + z]$ by Invariant \ref{invariant:2EFX_for_N_L} and the last inequality holds since $p(X_j^{t}) \in [k, k + z]$ by Invariant \ref{invariant:2pEFX_for_N_H}. Therefore we must have $i\in N^{t}_H$ and $j\in N^{t}_L$.
    
    Next we show that $X^{t}_j \subseteq \MPB_i$.
    
    First, observe that we must have $\alpha_i = 1$, as otherwise $\alpha_i = 1/k$ and we have $c_i(e) = 1$ and $p(e) = k$ for all $e\in X^{t}_i$.
    It implies that $|X^{t}_i|=1$, which is a contradiction with $i$ being not $(2 - 1/k)$-EFX towards $j$.
    Second, we show that $c_i(e) = 1$ for all $e\in X^{t}_j$.
    Suppose otherwise, then we have
    \begin{equation*}
        c_i(X^{t}_i) - 1 = \alpha_i \cdot p(X_i^{t}) - 1 \leq k + z - 1 < 2k - 1 \leq (2 - \frac{1}{k})\cdot c_i(X^{t}_j),
    \end{equation*}
    where the last inequality holds from the fact that $c_i(X_j^{t}) \geq k$. It's a contradiction with $i$ being not $(2 - 1/k)$-EFX towards $j$.
    Therefore we have $\alpha_{i,e} = 1$ for all $e\in X^{t}_j$, which implies $X^{t}_j \subseteq \MPB_i$.
\end{proof}

\begin{lemma} \label{lemma:invariants_hold}
    Invariants~\ref{invariant:1kequilibrium}, \ref{invariant:2EFX_for_N_L} and \ref{invariant:2pEFX_for_N_H} are maintained at the end of round $t$.
\end{lemma}
\begin{proof}
    In round $t$, we identify two agents $i\in N^t_H$, $j \in N^t_L$, and exchange some items.
    By Lemma~\ref{lemma:property_of_not2EFX} and~\ref{lemma:1k_equilibrium_properties}, we can guarantee that both agents receive items that are in their MPB set.
    Hence $(\bX^t, \bp)$ is an equilibrium.    
    As our algorithm does not change the payment of any item, the equilibrium remains $\{1,k\}$-payment.
    Hence Invariant~\ref{invariant:1kequilibrium} is maintained.

    Next, we show that Invariant~\ref{invariant:2EFX_for_N_L} is maintained.
    We remark that there is only one high payment item in $X^t_i$ since $p(X^t_i) \leq k+z < 2k$, which implies that agent $i$ only holds low payment items at the end of round $t$.
    Since in round $t$ only $i$ and $j$ exchange items, it suffices to show that at the end of round $t$, (1) agent $i$ (who joins $N^{t + 1}_L$) is $(2 - 1/k)$-EFX towards all other agents in $N^{t + 1}_L$; (2) all other agents in $N^{t + 1}_L\setminus \{i\}$ is $(2 - 1/k)$-EFX towards $i$; (3) $p(X_i^{t + 1})\in [z,k+z]$.

    \begin{enumerate}
        \item[(1)] For any agent $l\in N_L^t \setminus\{i\}$, if agent $i$ is $(2 - 1/k)$-EFX towards agent $l$ in allocation $\bX^{t}$, then we have $c_i(X^{t}_i) - 1 \leq (2 - 1/k) \cdot c_i(X^{t}_l)$. Since $\alpha_i = 1$ and $X^{t}_j \subseteq \MPB_i$ by Lemma \ref{lemma:property_of_not2EFX}, we have $c_i(X_j^{t}) = \alpha_i \cdot p(X_j^{t}) < k$ by Lemma~\ref{lemma:payment<k}.
        Hence after exchanging a high payment item with $X^t_j$, the bundle cost of agent $i$ does not increase, and therefore agent $i$ is still $(2 - 1/k)$-EFX towards $l$ (whose bundle does not change) at the end of round $t$.
        
        Now suppose that agent $i$ is not $(2 - 1/k)$-EFX towards agent $l$ in allocation $\bX^{t}$.
        By Lemma \ref{lemma:payment<k} we have $p(X_l^{t}) < k$.
        Recall that $j$ is the agent with the minimum earning such that $i$ is not $(2 - 1/k)$-EFX towards. 
        Hence $p(X^{t}_j) \leq p(X^{t}_l) = p(X^{t + 1}_l)$.
        Following Lemma~\ref{lemma:size-of-lowpaymentitems}, we have 
        \begin{equation*}
            p(X^{t + 1}_i) - 1 = p(X^{t}_i \cap L) + p(X^{t}_j) -1 \leq 2 \cdot p(X^{t}_j) - 1 \leq 2 \cdot p(X^{t + 1}_l) - 1 < (2 - \frac{1}{k}) \cdot p(X^{t + 1}_l),
        \end{equation*}
        where the first inequality holds since the allocation $\bX^t$ is pEF1 for agent $i$ and the last inequality holds from the fact that $p(X_l^{t + 1}) = p(X_l^{t}) < k$.
        
        
        \item[(2)] Following Invariant~\ref{invariant:2EFX_for_N_L}, any agent $l\in N_L^{t + 1} \setminus\{i\}$ is $(2 - 1/k)$-EFX towards agent $j$ in the allocation $\bX^{t}$.
        Since $X^{t}_j \subseteq X^{t + 1}_i$, agent $l$ is also $(2 - 1/k)$-EFX towards $i$ at the end of round $t$.
        
        \item[(3)] At the end of round $t$ we have 
        \begin{equation*}
            p(X^{t + 1}_i) = p(X^{t}_i \setminus H) + p(X^{t}_j) \leq (k + z - k) + p(X^{t}_j) < z + k,
        \end{equation*}
        where the first inequality holds following $p(X^{t}_i) \leq k+z$ and $|X_i^{t} \cap H| = 1$, the second inequality holds from the fact that $p(X^{t}_j) < k$ by Lemma~\ref{lemma:payment<k}. On the other hand, we have $p(X^{t + 1}_i) \geq p(X^{t}_j) \geq z$. In conclusion, $p(X_i^{t + 1}) \in [z, k + z]$.
    \end{enumerate}
    \smallskip

    Finally, we show that Invariant~\ref{invariant:2pEFX_for_N_H} is maintained. 
    It's clear that at the end of round $t$, $X_j^{t + 1}$ only contains an item with high payment, i.e., $p(X_j^{t + 1}) = k \in [k,k+z]$.
    For any $l \in N_H^{t + 1} \setminus \{j\}$, we have $p(X_{l}^{t}) = p(X_{l}^{t + 1})$, which implies that Invariant~\ref{invariant:2pEFX_for_N_H} is maintained at the end of round $t$.
    %
\end{proof}

Next, we will analyze the time complexity of the algorithm and demonstrate that it terminates in polynomial time.

\begin{lemma}
    The algorithm computes a $(2 - 1/k)$-EFX and PO allocation in $O(n^2m)$ time.
\end{lemma}
\begin{proof}
    By the above analysis, in each while-loop of the algorithm, an agent $i\in N^{t}_H$ and an agent $j\in N^{t}_L$ will be selected, and after item reallocations, agent $j$ receives one item with payment $k$ and agent $i$ receives all other low payment items.
    Therefore the number of agents in $N_H$ that receive at least one low payment item will decrease by one after each while-loop, which implies that the total number of while-loops is at most $n$.
    Given agent $i$, finding the agent $j$ that is most envied by agent $i$ can be done in $O(m)$ time. Moreover, checking whether the allocation is $(2 - 1/k)$-EFX can be done in $O(n m)$ time.
    Hence the algorithm runs in $O(n^2 m)$ time.
\end{proof}

In conclusion, for any given $\{1,k\}$-instance, we first compute the pEF1 $\{1,k\}$-payment equilibrium in $O(k n^2 m^2)$ time~\cite{conf/sigecom/WuZZ23}. 
Subsequently, we apply Algorithm~\ref{alg:reallocation_bivalued} to iteratively reallocate the items until the allocation is $(2 - 1/k)$-EFX.
Since we maintain an equilibrium, PO is guaranteed.
As demonstrated in the previous sections, our algorithm guarantees this outcome in polynomial time, leading to the following main result.

\begin{theorem} \label{theorem:2EFX_PO_bivalued}
There exists a polynomial-time algorithm that computes a $(2 - 1/k)$-EFX and PO allocation for every given bi-valued instance.
\end{theorem}

\section{EFX and PO Allocations for \texorpdfstring{$\{1,2\}$}{}-Instances}

In this section, we consider the case that $k=2$ and propose Algorithm \ref{alg:efxandpo} for computing EFX and PO allocations. As in Section~\ref{Section3}, our algorithm begins with the pEF1 $\{1,k\}$-payment equilibrium $(\bX, \bp)$ for $\{1,k\}$-instances by existing works~\cite{conf/aaai/GargMQ22,conf/sigecom/WuZZ23}\footnote{In this paper, we mainly follow the algorithm of~\cite{conf/sigecom/WuZZ23}, which originally works for the weighted setting in which each agent has a non-negative weight $w_i > 0$ and returns a pWEF1 equilibrium. Here we let all agents' weight be $1/n$.}.
However, in order to further improve the approximation guarantee (regarding EFX), we need more structural properties.
In the following, we review the algorithm of~\cite{conf/sigecom/WuZZ23}, highlighting some of their key operations and invariants, which will be the building blocks of our algorithm for computing EFX and PO allocations.

\subsection{Computation of pEF1 Equilibrium by \texorpdfstring{\cite{conf/sigecom/WuZZ23}}{}}

Their algorithm has two main components: (1) the initialization of an equilibrium $(\bX^0, \bp^0)$ (which is not necessarily pEF1) and partition of agents into groups $\{N_r\}_{r \in [R]}$; and (2) a sequence of item reallocations to turn the equilibrium pEF1.

\subsubsection{Initial Equilibrium and Agent Groups}

We first partition the items depending on whether they cost $k$ to all agents.

\begin{definition} [Item Groups]
We divide the items into two groups $M^+$ and $M^-$, where $M^+$ contains all consistently large items and $M^-$ contains the other items:
$$
    M^+ = \{e \in M : \forall i \in N, c_i(e) = k\}, \quad
    M^- = \{e \in M : \exists i \in N, c_i(e) = 1\}.
$$
\end{definition}

\paragraph{Initial Equilibrium.}
We set the payment vector $\bp^0$ as $p^0(e) = 1$ for all $e\in M^-$ and $p^0(e) = k$ for all $e\in M^+$.
Then for each item $e \in M^-$, we allocate it to an arbitrary agent $i$ with $c_i(e) = 1$; items in $M^+$ are allocated arbitrarily. 
Based on the allocation, we construct a directed graph $G_X$ in which the agents are vertices and there is an edge from agent $i$ to agent $j$ if $X_j\cap \MPB_i \neq \emptyset$.
As long as there is an MPB path from an agent $i$ to agent $j$ with $\min_{e\in X_i} \{p^0(X_i-e)\} > p^0(X_j)$, we resolve the path by reallocating items in reverse direction following the path.
We use $\bX^0$ to denote the eventual allocation and $(\bX^0, \bp^0)$ to denote the initial equilibrium (since we allocate an item $e$ to agent $i$ only if $e\in \MPB_i$, $(\bX^0, \bp^0)$ is an equilibrium).
Moreover, we have $\alpha_i = 1$ for all $i\in N$.

\begin{definition}[Big and Least Earner] \label{def:big-least-earners}
    Given an equilibrium $(\bX, \bp)$, an agent $b\in N$ is called a big earner if $b \in \arg\max_{i\in N} \left\{\hat{p}_i\right\}$, where $\hat{p}_i = \min_{e\in X_i}\{ p(X_i - e) \}$ is the earning of agent $i$ after removing the item with highest payment; an agent $l$ is called a least earner if $l \in \arg\min_{i\in N} \{ p(X_i) \}$.
\end{definition}
 
\paragraph{The Agent Groups.} 
Based on the initial allocation, we divide agents into agent groups $N_1, \ldots, N_R$ as follows: we repeatedly identify the big earner $b$ among agents without a group and put $b$ together with all agents that can reach $b$ on graph $G_X$ into one group.
We call $b$ the representative of the group.

\begin{lemma}[Agent Groups]\label{lemma:leximin}
    Based on the initial equilibrium $(\bX^0, \bp^0)$, we can partition the agents into groups $\{N_r\}_{r \in [R]}$ with the following properties:
    \begin{enumerate}
        \item[(1)] For all $1\leq r < r'\leq R$, agent $i\in N_r$ and $j\in N_{r'}$, we have $c_j(e) = k$ for all $e \in X_i^0$.
        \item[(2)] All consistently large items are allocated to the lowest group: $M^+ \subseteq \bigcup_{i \in N_R}X^0_i$.
        \item[(3)] Any two agents from the same group are pEF1 towards each other.
    \end{enumerate}
\end{lemma}

In the following, we call $N_1$ the \emph{highest} group and $N_R$ the \emph{lowest}.
Moreover, we re-index the agents such that agents in $N_1$ are assigned indices $\{1, 2, \dots, |N_1|\}$, agents in $N_2$ are assigned the indices $\{|N_1| + 1, |N_1| + 2, \dots, |N_1| + |N_2|\}$, and so on.

\subsubsection{Reallocations to Ensure pEF1}

Given the initial equilibrium $(\bX^0, \bp^0)$, we identify the big earner $b$ and the least earner $l$\footnote{Throughout the whole paper, when selecting a big earner or a least earner, we break ties by picking the one with the smallest index.}.
If $b$ is pEF1 toward $l$, then the equilibrium is pEF1; otherwise, we reallocate an item from $b$ to $l$, during which we may raise the payment of all items received by agents in the same group as $b$ by a factor of $k$. 

\begin{theorem}[\cite{conf/sigecom/WuZZ23}]\label{the:ef1+po-bivalued}
    Algorithm~\ref{alg:ef1andpo} computes an EF1 and PO allocation for any given bi-valued instance in polynomial time.
\end{theorem}

\begin{algorithm}[H]
    \caption{Compute a pEF1 equilibrium} \label{alg:ef1andpo}
    \SetKw{Break}{Break}
    \KwIn{A bi-valued instance $(M, N, \bf{c})$}
    let $(\bX, \bp, \{N_r\}_{r\in [R]})$ be the initial equilibrium and agent groups \;
    initialize the set of unraised agents $U\gets N$ \;
    \While{\textbf{True}}{
        let $b \gets \arg\max_{i\in N} \{ \hat{p}_i \}$ be the big earner \;
        let $l \gets \arg\min_{i\in N} \left\{ p(X_i) \right\}$ be the least earner \;
        \If{$\hat{p}_b \leq p(X_l)$}{
            \textbf{output} $(\bX, \bp)$ and \textbf{terminate}.
        }
        suppose that $b \in N_r$ and $l\in N_{r'}$\;
        \uIf {$l \in U$}{
            \If{$b \in U$}{
                raise the payment of all chores in $\bigcup_{i\in N_r} X_i$ by a factor of $k$ \;
                $U \gets U \setminus N_r$ \;
            }
            pick an arbitrary $e\in X_b$ \;
            update $X_b \gets X_b - e$, $X_l \gets X_l + e$ \;}
        \Else {
            we have $b\in N\setminus U$ and there exists $i \in U$ with $X_i \cap X^0_l \neq \emptyset$ \;
            pick any $e_1\in X_b \cap \MPB_i$ and $e_2\in X_i\cap X^0_l$ \;
            update $X_b \gets X_b - e_1$, $X_l \gets X_l + e_2$ and $X_i \gets X_i + e_1 - e_2$ \;
        }
    }
\end{algorithm} 



\begin{lemma} \label{lemma:Group_Invariant}
    During Algorithm~\ref{alg:ef1andpo}, following properties are always maintained:
    \begin{enumerate}
        \item[(1)] At any point in time, $(\bX, \bp)$ is an equilibrium.
        \item[(2)] For all $r\in [R]$, agents in the same group $N_r$ are pEF1 towards each other.
        \item[(3)] For all $i\in U$, we have $\alpha_i = 1$; for all $i\in N\setminus U$, we have $\alpha_i = 1/k$.
        \item[(4)] There exists $r^* < R$ such that all agents in groups $N_1,N_2,\ldots,N_{r^*}$ are raised while agents in groups $N_{r^*+1},N_{r^*+2},\ldots,N_{R}$ are unraised.
    \end{enumerate}
\end{lemma}
	
We use $\bX^t$ and $\bp^t$ to denote the allocation and the payment at the beginning of round $t$. Besides, we use $b^t$ and $l^t$ to denote the big earner and the least earner identified at the beginning of round $t$. 

\begin{lemma} \label{lemma:Earning_Invariant}
    The earning status of agents satisfies the following properties:
    \begin{enumerate}
        \item[(1)] For all $t$ we have $p^0(X^0_{l^0}) \leq p^1(X^1_{l^1})\leq \cdots \leq p^t(X^t_{l^t})$.
        \item[(2)] In each round, the earning of agent $i$ does not change, unless $i$ is raised during this round or $i$ is the big or least earner in this round. Moreover, for all $e\in M$, $p(e)\in \{1,k\}$.
    \end{enumerate}
\end{lemma}

\subsection{The Reallocation Algorithm}
In this section, we describe how to compute EFX and PO allocations for the case when $k = 2$, based on the pEF1 equilibrium ($\bX, \bp$) computed by Algorithm \ref{alg:ef1andpo}.
As in Section~\ref{Section3}, we use $H$ and $L$ to denote the high and low payment items, and $N_H$ and $N_L$ to denote the set of agents with and without high payment items.
By Property (3) of Lemma~\ref{lemma:Group_Invariant}, all raised agents are in $N_H$, and those agents have only high payment items.
Similar to Lemma~\ref{lemma:1k_equilibrium_properties}, for all agent $i\in U$ (unraised agents) we have $\alpha_i = 1$, which implies $c_i(e) = k$ and $e\in \MPB_i$ for all $e\in H$.

\begin{lemma}\label{lemma:k-in-MPB-U}
    For any item $e\in H$ and agent $i\in U$, we have $c_i(e) = k$ and $e\in \MPB_i$.
\end{lemma}

Note that for $k=2$, we have $| p(X_i) - p(X_j) | \leq 2$ for any two agents $i,j$ since $(\bX,\bp)$ is pEF1.
Hence, we can further partition the agents into three groups $N^z, N^{z+1}$ and $N^{z+2}$ for some integer $z$, where an agent $i$ is contained in $N^t$ if $p(X_i) = t$.
Moreover, if $N^z = \emptyset$ or $N^{z+2} = \emptyset$, then the allocation is already pEFX (thus EFX and PO).
Therefore, a natural idea is to reallocate items between agents in $N^z$ and $N^{z+2}$ until one of the two groups is empty.
However, due to similar difficulties we mentioned in the previous section, achieving pEFX could be very difficult.
We therefore focus on achieving EFX while maintaining the equilibrium property.
We show that as long as the allocation is not EFX, we can find an agent $i\in N^{z+2}$ that is not EFX towards an agent $j\in N^z$, and reallocate some items between agent $i$ and agent $j$, after which either both of them join $N^{z+1}$, or one of them receives only high payment items and will never participate in any item reallocation in the future.
Throughout the reallocation, we maintain the following invariant.

\begin{invariant}[Payment Invariant]\label{invariant:EFX_payment}
    For any agent $i\in N$, $p(X_i) \in \{z, z+1, z+2\}$. Moreover, $(\bX, \bp)$ is always a pEF1 equilibrium.
\end{invariant}

To maintain the above invariant, we ensure that when we reallocate an item $e$ to an agent $i$, we have $e\in \MPB_i$.
Moreover, we ensure $N^{z+2} \subseteq N_H$ so that the equilibrium is pEF1.
Note that throughout the whole reallocation procedure, we do not change the payment of any item, which means that for any raised agent $i$ we always have $\alpha_i = \frac{1}{k}$.
Therefore, as long as $(\bX, \bp)$ is an equilibrium, all raised agents receive only high payment items and the equilibrium is pEFX for them, because their earnings after removing any item is at most $z$.

\smallskip

In the following, we establish some properties for agents $i$ and $j$, if $i$ is not EFX towards $j$. These properties will play a crucial role in designing our reallocation algorithm.

\begin{lemma} \label{lemma:property_of_notEFX}
    If agent $i$ is pEF1 towards agent $j$ but not EFX towards $j$, then we have
    \begin{itemize}
        \item $i\in N^{z+2}$ and $j\in N^z$;
        \item $X_i\cap H\neq \emptyset$ and $X_i\cap L\neq \emptyset$, which implies that agent $i\in N_H\cap U$;
        \item $X_j \subseteq \MPB_i$;
        \item $j\in U$. 
    \end{itemize}
\end{lemma}
\begin{proof}
    The first statement is trivial because otherwise, $i$ is pEFX towards $j$.
    For the second statement, observe that we must have $\alpha_i = 1$, as otherwise ($\alpha_i = \frac{1}{k}$) we have $c_i(e) = 1$ and $p(e) = k$ for all $e\in X_i$, which implies that $i$ is EFX towards $j$ since all items in $X_i$ have the same cost.
    By Property (3) of Lemma~\ref{lemma:Group_Invariant}, we have $i\in U$.
    Since $i$ is pEF1 towards $j$ but not EFX (thus not pEFX) towards $j$, $X_i$ must contain both items with payment $1$ and payment $k$.
    Hence the second statement follows.
    
    Furthermore, we have $c_i(e) = \alpha_{i,e}\cdot p(e) \geq \alpha_i\cdot p(e) = p(e)$ for all $e\in X_j$.
    Moreover, if there exists any $e\in X_j$ such that $c_i(e) > p(e)$, then we have $c_i(e) = k$ and $p(e) = 1$, which implies
    \begin{equation*}
        c_i(X_j) \geq p(X_j) + k-1 \geq p(X_i) - 1 = c_i(X_i) - 1 \geq \max_{e\in X_i} \{c_i(X_i - e)\},
    \end{equation*}
    where the second inequality holds since $i$ is pEF1 towards $j$.
    Hence we have a contradiction with $i$ being not EFX towards $j$, which implies that $\alpha_{i,e} = 1$ for all $e\in X_j$ and $X_j \subseteq \MPB_i$.

    Finally, we show that $j$ is unraised.
    Assume the contrary that $j$ is raised, we have $\alpha_j = 1/k$, which implies $c_j(e) = 1$ and $p(e) = k$ for all $e\in X_j$.
    Recall that $i$ is unraised, which implies that $i$ has never lost any item during Algorithm~\ref{alg:ef1andpo} (only raised agents give out items).

    We first consider the case that $i$ has never received any item during Algorithm~\ref{alg:ef1andpo}.
    Note that $X_i \cap H \neq \emptyset$ and $X_i = X^0_i$.
    Property (2) of Lemma~\ref{lemma:leximin} implies that $i \in N_R$.
    By the way we partition the agents into groups, the earning of representative agents in groups $N_1,N_2,\ldots,N_R$ are non-increasing, together with the fact that agents in the same group are pEF1 towards each other (see Property (2) of Lemma~\ref{lemma:Group_Invariant}), we obtain $p^0(X_{l^0}^0) \geq p(X_i) - k = z$.
    Then by Property (1) of Lemma~\ref{lemma:Earning_Invariant}, we have $p^t(X_{l^t}^t) \geq z$ for all $t$, i.e., the earning of the least earner in every round is at least $z$.
    \begin{itemize}
        \item If $j$ has lost some items, then for the last time $j$ loses some item, $j$ must be the big earner in some round $t$ and the earning of $j$ decreased from $z + k$ to $z$. However, this would imply $p^t(X_{l^t}^t) < z$, which contradicts the conclusion we draw above.
        
        \item If $j$ has never lost any item, assuming that $j \in N_r$, there exists a round $t$ such that agents in group $N_r$ raise their earnings. Since $j$ never loses any item, in round $t$, the earning of $j$ increased from $\frac{z}{k}$ to $z$. Then we have 
        $$
        z = p^t(X_i^t) - k = \hat{p}^t_i \leq \hat{p}^t_{b^t} \leq p^t(X_j^t) = \frac{z}{k},
        $$
        where the first inequality holds since $b^t$ is the big earner in round $t$, and the second inequality holds since agents in $N_r$ are pEF1 towards each other. Therefore, we have a contradiction.
    \end{itemize}
    
    Next, we consider the case when $i$ has received some items.
    We consider the last round $t_i$ that the earning of agent $i$ increased (from $z$ to $z + k$).
    Let $t_j$ denote the last round that the earning of $j$ changed, we show that $t_j < t_i$ under three cases:
    \begin{itemize}
        \item At round $t_j$, the earning of $j$ decreased from $z + k$ to $z$. 
        In this case, we have $p^{t_j}(X_{l_{t_j}}^{t_j}) < z$, i.e., the earning of the least earner in round $t_j$ is less than $z$.
        Since the earning of the least earner is non-decreasing by Property (1) of Lemma~\ref{lemma:Earning_Invariant}, we have $t_j < t_i$.
        \item At round $t_j$, the earning of $j$ increased from $z - k$ to $z$. 
        In other words, agent $j$ is the least earner in round $t_j$.
        Then by Property (1) of Lemma~\ref{lemma:Earning_Invariant} and same argument as above, we have $t_j < t_i$.
        \item 
        {At round $t_j$, agent $j$ is raised and her earning increase from $z/k$ to $z$. Thus the least earner in round $t_j$ has earning at most $z/k < z$. By the same argument as above we also have $t_j < t_i$. }
    \end{itemize}
    
    Therefore, we can conclude that $t_i > t_j$. 
    By the definition of $t_j$, the earning of agent $j$ is $z$ at the beginning of round $t_i$.
    Moreover, since agent $j$ is raised and agent $i$ is not, we have $j < i$ (by Property (4) of Lemma~\ref{lemma:Group_Invariant}).
    However, this implies that in round $t_i$, agent $i$ should not be selected as the least spender, as we break by choosing agents with a smaller index. 
    Therefore we have a contradiction.
\end{proof}

We now proceed to introduce the algorithm for reallocating items. 
The design of our algorithm heavily relies on Lemma~\ref{lemma:property_of_notEFX}, which establishes that if agent $i$ is not EFX towards $j$, then both $i$ and $j$ must be unraised agents possessing nice properties. 

\smallskip

\begin{algorithm}[H]
    \caption{Reallocation algorithm for $\{1, 2\}$-instances} \label{alg:efxandpo}
    \SetKw{Break}{Break}
    \KwIn{pEF1 equilibrium $(\bX, \bp)$ computed by Algorithm \ref{alg:ef1andpo}, set of unraised agents $U \subseteq N$}
    \While{there exists agent $i$ that is not EFX towards agent $j$}{
        \tcp{we have $i \in N^{z+2}\cap U$ and $j\in N^z\cap U$ by Lemma~\ref{lemma:property_of_notEFX}}
        pick an arbitrary $e_i \in X_i$ s.t. $p(e_i) = k$; \qquad\qquad\quad \tcp{$e_i \in \MPB_j$, by Lemma~\ref{lemma:k-in-MPB-U}}
        \If{there exists an item $e_j \in X_j$  such that $p(e_j) = 1$}
        {
            pick $e_j \in X_j$ such that $p(e_j) = 1$;    \qquad\qquad\qquad \tcp{$e_j \in \MPB_i$, by Lemma~\ref{lemma:property_of_notEFX}}
        $X_i \gets X_i + e_j - e_i$, $X_j \gets X_j + e_i - e_j$.
        }
        \Else {
            $X_i \gets X_i - e_i$, $X_j \gets X_j + e_i$.
        }
    }
\end{algorithm} 

For an illustration for the execution of the algorithm, see Example~\ref{example1:execution_of_12reallocation_1} and Example~\ref{example1:execution_of_12reallocation_2}.

\begin{example}\label{example1:execution_of_12reallocation_1}
    Consider the following $\{1, 2\}$-instance with six agents. 
    Before the reallocation, the payment of the bundles are as follows: $p(X_1) = 4$, $p(X_2) = 5$, $p(X_3) = 5$, $p(X_4) = 5$, $p(X_5) = 6$, and $p(X_6) = 6$. 
    In this instance, agent $5$ holds both high and low payment items. Suppose that agent $5$ is not EFX towards agent $1$, whose bundle also contains items with payments of $1$ and $2$. The algorithm identifies items $e_5 \in X_5$ and $e_1 \in X_1$ such that $p(e_5) = 2$ and $p(e_1) = 1$. 
    Then, agent $1$ receives $e_5$, and agent $5$ receives $e_1$. See Figure~\ref{figure:exchange-1-2} for an illustration of this example.
\end{example}

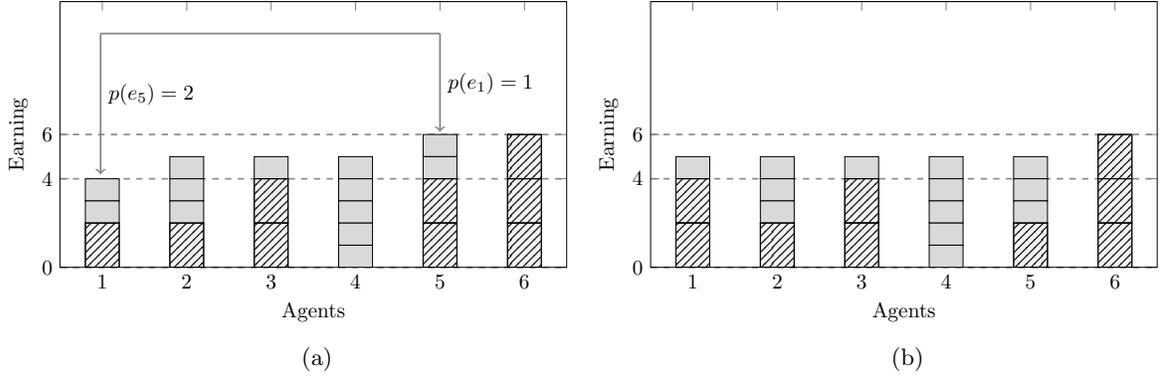
\begin{figure}[htb]
    \centering
    \begin{subfigure}[c]{0.45\textwidth} 
        \centering
        \begin{tikzpicture}[scale=0.8]
        \begin{axis}[
            ybar stacked,
            ymax=12,
            ymin=0,
            enlarge x limits=0.1,
            enlarge y limits=false,
            bar width=16pt,
            height=6cm, width=10cm,
            ylabel={Earning},
            ytick={0, 4, 6},
            xlabel={Agents},
            symbolic x coords={1, 2, 3, 4, 5, 6}, 
            xtick=data,
            xticklabels={1, 2, 3, 4, 5, 6},
            ymajorgrids=true,
            grid style={dashed, line width=0.8pt, gray!90},
        ]
        \addplot+[draw=black,postaction={pattern=north east lines}, fill=gray!10] coordinates {
            (1,2) (2,2) (3,0) (4,0) (5,2) (6,2)
        };
        \addplot+[draw=black,postaction={pattern=north east lines}, fill=gray!10] coordinates {
            (1,0) (2,0) (3,2) (4,0) (5,2) (6,2)
        };
        \addplot+[draw=black,postaction={pattern=north east lines}, fill=gray!10] coordinates {
            (1,0) (2,0) (3,2) (4,0) (5,0) (6,2)
        };
        \addplot+[draw=black, fill=gray!30] coordinates {
            (1,0) (2,1) (3,1) (4,1) (5,0) (6,0)
        }; 
        \addplot+[draw=black, fill=gray!30] coordinates {
            (1,0) (2,0) (3,0) (4,1) (5,1) (6,0)
        }; 
        \addplot+[draw=black, fill=gray!30] coordinates {
            (1,0) (2,1) (3,0) (4,1) (5,1) (6,0)
        }; 
        \addplot+[draw=black, fill=gray!30] coordinates {
            (1,1) (2,1) (3,0) (4,1) (5,0) (6,0)
        }; 
        \addplot+[draw=black, fill=gray!30] coordinates {
            (1,1) (2,0) (3,0) (4,1) (5,0) (6,0)
        }; 

        \node[anchor=north] at (rel axis cs:0.31, -0.18) {\(N_L^{t}\)};
        \node[anchor=north] at (rel axis cs:0.80, -0.18) {\(N_H^{t}\)};
        
        \draw[<-, thick, gray] (rel axis cs:0.08,0.35) -- (rel axis cs:0.08,0.88);
        \draw[-, thick, gray] (rel axis cs:0.08,0.88) -- (rel axis cs:0.75,0.88);
        \draw[->, thick, gray] (rel axis cs:0.75,0.88) -- (rel axis cs:0.75,0.51);
        \node[anchor=north] at (rel axis cs:0.85,0.76) {$p(e_1) = 1$};
        \node[anchor=north] at (rel axis cs:0.18, 0.72) {$p(e_5) = 2$};
        \end{axis}
        \end{tikzpicture}
        \caption{\hspace{-1.47cm}     }
    \end{subfigure}
    \hspace{0.6cm}
    \begin{subfigure}[c]{0.45\textwidth} 
        \centering
        \begin{tikzpicture}[scale=0.8]
        \begin{axis}[
            anchor=north,
            yshift=-3cm,
            ybar stacked,
            ymax=12,
            ymin=0,
            enlarge x limits=0.1,
            enlarge y limits=false,
            bar width=16pt,
            height=6cm,
            width=10cm,
            ylabel={Earning},
            ytick={0, 4, 6},
            xlabel={Agents},
            symbolic x coords={1,2,3,4,5,6},
            xtick=data,
            xticklabels={1, 2, 3, 4, 5, 6},
            ymajorgrids=true,
            grid style={dashed, line width=0.8pt, gray!90}
        ]
        \addplot+[draw=black,postaction={pattern=north east lines}, fill=gray!10] coordinates {
            (1,2) (2,2) (3,2) (4,0) (5,2) (6,2)
        };
        \addplot+[draw=black,postaction={pattern=north east lines}, fill=gray!10] coordinates {
            (1,0) (2,0) (3,2) (4,0) (5,0) (6,2)
        };
        \addplot+[draw=black,postaction={pattern=north east lines}, fill=gray!10] coordinates {
            (1,0) (2,0) (3,0) (4,0) (5,0) (6,2)
        };

        \addplot+[draw=black,postaction={pattern=north east lines}, fill=gray!10] coordinates {
            (1,2) (2,0) (3,0) (4,0) (5,0) (6,0)
        };
        \addplot+[draw=black, fill=gray!30] coordinates {
            (1,0) (2,1) (3,0) (4,0) (5,1) (6,0)
        };

        \addplot+[draw=black, fill=gray!30] coordinates {
            (1,0) (2,0) (3,1) (4,1) (5,0) (6,0)
        }; 
        \addplot+[draw=black, fill=gray!30] coordinates {
            (1,0) (2,0) (3,0) (4,1) (5,1) (6,0)
        }; 
        \addplot+[draw=black, fill=gray!30] coordinates {
            (1,0) (2,1) (3,0) (4,1) (5,1) (6,0)
        }; 
        \addplot+[draw=black, fill=gray!30] coordinates {
            (1,0) (2,1) (3,0) (4,1) (5,0) (6,0)
        }; 
        \addplot+[draw=black, fill=gray!30] coordinates {
            (1,1) (2,0) (3,0) (4,1) (5,0) (6,0)
        }; 
        \end{axis}

        \end{tikzpicture}
        \caption{\hspace{-1.47cm}  }
    \end{subfigure}
    \caption[c]{The illustration of Example~\ref{example1:execution_of_12reallocation_1}. ($a$) In this instance, $z = 4$, with $N^z = \{1\}$, $N^{z + 1} = \{2,3,4\}$ and $N^{z + 2} = \{5,6\}$. ($b$) The earning status of agents after reallocation, where $p(X_1) = p(X_5) = 5$, i.e., both of them joins $N^{z+1}$ after the reallocation.}
    \label{figure:exchange-1-2}
\end{figure}

\begin{example}\label{example1:execution_of_12reallocation_2}
    Consider another $\{1, 2\}$-instance with six agents where $p(X_1) = 4$, $p(X_2) = 4$, $p(X_3) = 4$, $p(X_4) = 5$, $p(X_5) = 6$, and $p(X_6) = 6$. 
    Again, agent $5$ holds both high and low payment items, and suppose that agent $5$ is not EFX towards agent $1$, whose bundle contains only high payment items. 
    The algorithm will select an item $e_5 \in X_5$ such that $p(e_5) = 2$, and this item will be reallocated to agent $1$. 
    See Figure~\ref{figure:reallocate-2} for an illustration of this example.
\end{example}

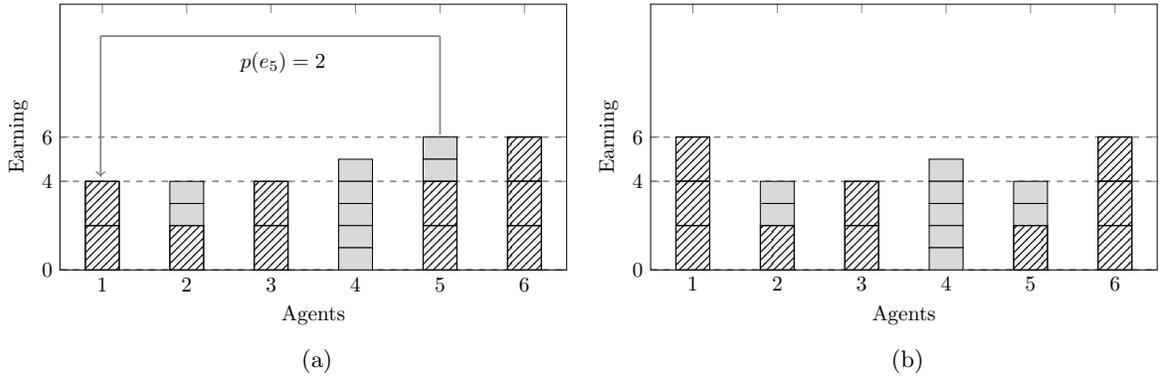
\begin{figure}[htb]
    \centering
    \begin{subfigure}[c]{0.45\textwidth} 
        \centering
        \begin{tikzpicture}[scale=0.8]
        \begin{axis}[
            ybar stacked,
            ymax=12,
            ymin=0,
            enlarge x limits=0.1,
            enlarge y limits=false,
            bar width=16pt,
            height=6cm, width=10cm,
            ylabel={Earning},
            ytick={0, 4, 6},
            xlabel={Agents},
            symbolic x coords={1, 2, 3, 4, 5, 6}, 
            xtick=data,
            xticklabels={1, 2, 3, 4, 5, 6},
            ymajorgrids=true,
            grid style={dashed, line width=0.8pt, gray!90},
        ]
        \addplot+[draw=black,postaction={pattern=north east lines}, fill=gray!10] coordinates {
            (1,2) (2,2) (3,2) (4,0) (5,2) (6,2)
        };
        \addplot+[draw=black,postaction={pattern=north east lines}, fill=gray!10] coordinates {
            (1,2) (2,0) (3,2) (4,0) (5,2) (6,2)
        };
        \addplot+[draw=black,postaction={pattern=north east lines}, fill=gray!10] coordinates {
            (1,0) (2,0) (3,0) (4,0) (5,0) (6,2)
        };
        \addplot+[draw=black, fill=gray!30] coordinates {
            (1,0) (2,0) (3,0) (4,1) (5,0) (6,0)
        }; 
        \addplot+[draw=black, fill=gray!30] coordinates {
            (1,0) (2,0) (3,0) (4,1) (5,1) (6,0)
        }; 
        \addplot+[draw=black, fill=gray!30] coordinates {
            (1,0) (2,1) (3,0) (4,1) (5,1) (6,0)
        }; 
        \addplot+[draw=black, fill=gray!30] coordinates {
            (1,0) (2,1) (3,0) (4,1) (5,0) (6,0)
        }; 
        \addplot+[draw=black, fill=gray!30] coordinates {
            (1,0) (2,0) (3,0) (4,1) (5,0) (6,0)
        }; 

        \node[anchor=north] at (rel axis cs:0.31, -0.18) {\(N_L^{t}\)};
        \node[anchor=north] at (rel axis cs:0.80, -0.18) {\(N_H^{t}\)};
        
        \draw[<-, thick, gray] (rel axis cs:0.08,0.35) -- (rel axis cs:0.08,0.88);
        \draw[-, thick, gray] (rel axis cs:0.08,0.88) -- (rel axis cs:0.75,0.88);
        \draw[-, thick, gray] (rel axis cs:0.75,0.88) -- (rel axis cs:0.75,0.51);
        \node[anchor=north] at (rel axis cs:0.44,0.85) {$p(e_5) = 2$};
        \end{axis}
        \end{tikzpicture}
        \caption{\hspace{-1.47cm}     }
    \end{subfigure}
    \hspace{0.6cm}
    \begin{subfigure}[c]{0.45\textwidth} 
        \centering
        \begin{tikzpicture}[scale=0.8]
        \begin{axis}[
            anchor=north,
            yshift=-3cm,
            ybar stacked,
            ymax=12,
            ymin=0,
            enlarge x limits=0.1,
            enlarge y limits=false,
            bar width=16pt,
            height=6cm,
            width=10cm,
            ylabel={Earning},
            ytick={0, 4, 6},
            xlabel={Agents},
            symbolic x coords={1,2,3,4,5,6},
            xtick=data,
            xticklabels={1, 2, 3, 4, 5, 6},
            ymajorgrids=true,
            grid style={dashed, line width=0.8pt, gray!90}
        ]
        \addplot+[draw=black,postaction={pattern=north east lines}, fill=gray!10] coordinates {
            (1,2) (2,2) (3,0) (4,0) (5,2) (6,2)
        };
        \addplot+[draw=black,postaction={pattern=north east lines}, fill=gray!10] coordinates {
            (1,2) (2,0) (3,2) (4,0) (5,0) (6,2)
        };
        \addplot+[draw=black,postaction={pattern=north east lines}, fill=gray!10] coordinates {
            (1,2) (2,0) (3,2) (4,0) (5,0) (6,2)
        };
        
        \addplot+[draw=black, fill=gray!30] coordinates {
            (1,0) (2,0) (3,0) (4,0) (5,0) (6,0)
        }; 
        
        \addplot+[draw=black, fill=gray!30] coordinates {
            (1,0) (2,0) (3,0) (4,1) (5,0) (6,0)
        }; 
        \addplot+[draw=black, fill=gray!30] coordinates {
            (1,0) (2,0) (3,0) (4,1) (5,1) (6,0)
        }; 
        \addplot+[draw=black, fill=gray!30] coordinates {
            (1,0) (2,1) (3,0) (4,1) (5,1) (6,0)
        }; 
        \addplot+[draw=black, fill=gray!30] coordinates {
            (1,0) (2,1) (3,0) (4,1) (5,0) (6,0)
        }; 
        \addplot+[draw=black, fill=gray!30] coordinates {
            (1,0) (2,0) (3,0) (4,1) (5,0) (6,0)
        }; 
        \end{axis}

        \end{tikzpicture}
        \caption{\hspace{-1.47cm}  }
    \end{subfigure}
    \caption[c]{The illustration of Example~\ref{example1:execution_of_12reallocation_2}. ($a$) The earning status of agents before reallocation. ($b$) The earning status of agents after reallocation, where agent $1$ receives only high payment items. }
    \label{figure:reallocate-2}
\end{figure}

\begin{lemma} \label{lemma:EFX_invariants_hold}
    Invariant~\ref{invariant:EFX_payment} is maintained throughout the whole algorithm.
\end{lemma}
\begin{proof}
Since the invariant holds for the pEF1 equilibrium computed by Algorithm~\ref{alg:ef1andpo}, it suffices to show that if Invariant~\ref{invariant:EFX_payment} holds at the beginning of some round $t$, then it also holds at the end of the round.
Note that the item reallocation only happens between agents $i, j$.
If they swap items following Line 3 - 5, then after the swap both agents have earning $z+1$ (see Figure~\ref{figure:exchange-1-2} for an example); if we reallocate an item following Line 7, then after the reallocation we have $p(X_i) = z$ and $p(X_j) = z+2$ (see Figure~\ref{figure:reallocate-2} for an example).
Thus, the first part of the invariant is maintained.
By Lemma~\ref{lemma:k-in-MPB-U} and~\ref{lemma:property_of_notEFX}, the new item received by an agent must be an MPB item.
Thus, $(\bX, \bp)$ is an equilibrium.
Moreover, it is maintained that every agent $i\in N^{z+2}$ holds at least one high payment item.
Thus, pEF1 is guaranteed.
\end{proof}

Given all lemmas and invariants, we can now prove the main result of this section.

\begin{theorem}\label{theorem:EFX_k=2}
    For any $\{1, 2\}$-instances with $n$ agents and $m$ indivisible chores, EFX and PO allocations exist and can be computed in polynomial time.
\end{theorem}
\begin{proof}
 Given any $\{1, 2\}$-instance, we first compute the pEF1 equilibrium in polynomial time by Theorem~\ref{the:ef1+po-bivalued}. 
 The while-condition of Algorithm \ref{alg:efxandpo} ensures that the returned allocation is EFX.
 The Pareto optimality is guaranteed by the equilibrium (see Invariant~\ref{invariant:EFX_payment}).
 In the following, we show that Algorithm \ref{alg:efxandpo} will terminate in polynomial time.
 Note that when agent $i$ is not EFX towards agent $j$, it must be the case that $i\in N^{z+2}\cap U$ and $j\in N^z\cap U$ by Lemma~\ref{lemma:property_of_notEFX}.
 In the following, we show that after each while-loop, the number of agents in $N^{z + 2}$ holding both high and low payment items strictly decreases, which implies that the algorithm terminates in $O(n)$ rounds.
 For convenience, we denote this number by $\eta$.
\begin{enumerate}
    \item If there exists an item $e_j \in X_j$ such that $p(e_j) = 1$, after swap operation we have $p(X_j) = p(X_i) = z + 1$ (in Line 3 - 5).
    Since $i$ joins $N^{z+1}$ after the swap, $\eta$ decreases by one.
    \item Otherwise, we have $p(e) = 2$ for all $e\in X_j$ and we reallocate one high payment item from $X_i$ to $X_j$ (in Line 7), after which agent $i$ joins $N^z$ and agent $j$ joins $N^{z+2}$.
    Since agent $i$ leaves $N^{z+2}$ and agent $j$ has only high payment items, $\eta$ decreases by one. 
\end{enumerate}

It is clear that each round of Algorithm \ref{alg:efxandpo} executes in $O(m)$ time. Therefore, the whole algorithm runs in polynomial time.
\end{proof}

\section{Conclusion and Open Problems}

In this work, we present a polynomial-time algorithm that computes $(2 - 1/k)$-EFX and PO allocations for $\{1, k\}$-instances, improving the state-of-the-art approximation ratio for EFX allocations for bi-valued instances. 
We also present a polynomial-time algorithm for the computation of EFX and PO allocations for $\{1,2\}$-instances.
Our results enrich and expand the growing literature on the computation and approximation of EFX allocations for chores.
However, it remains a fascinating open problem whether EFX allocations (even without the PO requirement) exist for bi-valued instances. 
Our algorithm and analysis framework (and that of Garg et al.~\cite{garg2024fair}) start from an equilibrium for chores and reallocate items maintaining MPB feasibility, which seem to have a great potential in improving other results regarding EF1/EFX and PO allocations for chores.
It would be an interesting open question to study which type of equilibrium and which specific properties of the equilibrium would be most helpful for approximating EF1/EFX and PO allocations.


\bibliography{main}
\bibliographystyle{abbrv}
\end{document}